\documentclass[a4paper,UKenglish,cleveref, autoref,thm-restate,authorcolumns]{lipics-v2019}

\graphicspath{{./figures/}}

\title{Parametrized Universality Problems for One-Counter Nets}

\titlerunning{}

\author{Shaull Almagor}{Technion, Israel}{}{https://orcid.org/0000-0001-9021-1175}{European Union's Horizon 2020 research and innovation programme under the Marie Sk{\l}odowska-Curie grant agreement No 837327.}

\author{Udi Boker}{Interdisciplinary Center (IDC) Herzliya, Israel}{}{}{Israel Science Foundation grant 1373/16.}

\author{Piotr Hofman}{University of Warsaw, Poland}{}{https://orcid.org/0000-0001-9866-3723}{Supported by the NCN grant 2017/27/B/ST6/02093.}

\author{Patrick Totzke}{University of Liverpool, United Kingdom}{}{http://orcid.org/0000-0001-5274-8190}{}

\authorrunning{S. Almagor, U. Boker,  P. Hofman, P. Totzke}

\Copyright{Shaull Almagor, Udi Boker,  Piotr Hofman, Patrick Totzke}



\category{}

\relatedversion{}

\supplement{}

\acknowledgements{
We are grateful for fruitful discussions during the Autoboz'2019 workshop.
}

\nolinenumbers 

\hideLIPIcs  

\EventEditors{John Q. Open and Joan R. Access}
\EventNoEds{2}
\EventLongTitle{42nd Conference on Very Important Topics (CVIT 2016)}
\EventShortTitle{CVIT 2016}
\EventAcronym{CVIT}
\EventYear{2016}
\EventDate{December 24--27, 2016}
\EventLocation{Little Whinging, United Kingdom}
\EventLogo{}
\SeriesVolume{42}
\ArticleNo{23}
\usepackage[basic,small]{complexity}
\usepackage{thm-restate}

\usepackage{calc}
\newcommand{\myNC}[1]{\NC\textsuperscript{#1}}
\newcommand{\myAC}[1]{\AC\textsuperscript{#1}}

\usepackage{xspace}
\usepackage{macros}

\begin{document}

\maketitle

\begin{abstract}
    We study the language universality problem for 
One-Counter Nets, also known as 1-dimensional Vector Addition Systems with States (1-VASS),
parameterized either with an initial counter value, or with an upper bound on the allowed counter value during runs.
The language accepted by an OCN (defined by reaching a final control state) is monotone in both parameters. This yields two natural questions:
1) does there exist an initial counter value that makes the language universal?
2) does there exist a sufficiently high ceiling so that the bounded language is universal?

Although the ordinary universality problem is decidable (and Ackermann-complete)
and these parameterized variants seem to reduce to checking basic structural properties of the underlying automaton, we show that in fact both problems are undecidable.
We also look into the complexities of the problems for several decidable subclasses,
namely for unambiguous, and deterministic systems, and for those over a single-letter alphabet.

\end{abstract}

\section{Introduction}
\label{sec:intro}
One-Counter Nets (OCNs) are finite-state machines equipped with an integer counter that cannot decrease below zero and which cannot be explicitly tested for zero.
They are the same as 1-dimensional Vector Addition Systems (or Petri nets with exactly one unbounded place).
In order to use them as formal language acceptors we assume that transitions are labelled with letters from a finite alphabet
and that some states are marked as accepting.

OCNs are a syntactic restriction of One-Counter Automata -- Minsky Machines with only one counter, which can have
zero-tests, i.e., transitions that depend on the counter value being exactly zero.
If counter updates are restricted to $\pm 1$, the model corresponds to Pushdown automata with a single-letter stack alphabet.
OCNs are one of the simplest types of discrete infinite-state systems, which makes them suitable for exploring the decidability border of classical decision problems from automata and formal-language theory.

\subparagraph*{Universality Problems.} 

The universality problem for a class of automata asks if a given automaton accepts all words over its input alphabet. 
%
Due to their lack of an explicit zero-test,
OCNs are monotone with respect to counter values:
if it is possible to make an $a$-labelled step from a configuration with state $p$ and counter $n$
to state $q$ with counter $n+d$, written as $(p,n) \step{a} (q,n+d)$ here,
then the same holds for any larger counter value $m\ge n$:
$(p,m) \step{a}(q,m+d)$.
Consequently, if we define the language via acceptance by reaching a final control state,
then for all states $s$ and $n\le m\in\N$,
the language 
$\Lang{s,n}$
of the initial configuration $(s,n)$
is included in that of $(s,m)$.
This motivates our first variation of the universality problem. 
The \emph{Initial-Value Universality} problem asks if there exists a sufficiently large initial counter to make the resulting language universal.

\dproblem
{
	An OCN with alphabet $\alphabet$
	and an initial state $\initialstate$. 
}{
	Does there exist $c_0\in\N$ such that $\Lang{\initialstate,c_0} = \alphabet^*$?
}

\noindent
The second question we consider is the \emph{Bounded Universality} problem,
which asks if there exists a large enough upper bound on the counter so that
every word can be accepted via a run that remains within this bound.
Writing $\bLang{\initialstate,c_0}{b}\subseteq \alphabet^*$ for the $b$-bounded language from configuration $(\initialstate,c_0)$, the decision problem is as follows.

\dproblem
{
	An OCN with alphabet $\alphabet$,
	an initial state $\initialstate$, and $c_0\in \N$.
}{
	Does there exist $b\in\N$ such that $\bLang{\initialstate,c_0}{b} = \alphabet^*$?
}
The motivation for studying these parameterized problems comes from the observation that
the ``vanilla'' universality problem, 
without existentially quantifying over parameters, 
is decidable, but Ackermann-complete \cite{HT2017}, and the lower bound depends strongly on the assumption that we start with a fixed initial counter (and that its value is not bounded). 
The two new variants of the universality problem relax these assumptions
in an attempt to allow efficient decision procedures via simple cycle analysis or similar.

\subparagraph*{Our Results.}
We show that both initial-value universality and bounded universality are undecidable (\cref{sec:nondet}).
The proofs use techniques from weighted automata \cite{DDGRT10,ABK11}, reducing the halting problem of two-counter machines to our setting. 

In light of these negative results, we proceed to study restricted classes of OCNs, for which the problems become decidable, as we elaborate below.
In most cases, the complexity crucially depends on how transition updates are encoded: we consider both the case of ``succinct'', binary-encoded updates, and the case of unary-encoded updates, which corresponds to systems where transitions can only update the counter by $\pm 1$.

The most intricate and interesting case is that of OCNs over a single-letter alphabet
(\cref{sec:unary}). In order to analyze this model, we split universality to criteria on ``short'' words, and on longer words that admit a cyclic behavior. In particular, we devise a canonical representation of ``pumpable'' paths, akin to the so-called linear-path schemes~\cite{LS2004,BFGHM15}. We show that the complexity of some of the problems is $\coNP$ complete, where others range between $\coNP$ and $\SigmaTwo$ (see \cref{tbl:Unary,tbl:Binary}).

We then consider deterministic, and unambiguous OCNs (\cref{sec:deterministic,sec:unambiguous}, respectively). For such systems, deciding (bounded) universality problems mostly reduces to checking
simple conditions on the cyclic structure of the control automaton underlying the OCN.
Based on known (but in some cases very recent) results on 
unambiguous finite automata and vector-addition systems,
we derive relatively low complexity upper bounds, in polynomial time (assuming unary encoding)
and space (assuming binary encoding).
\Cref{tbl:Unary,tbl:Binary} summarize the status quo, following our results.

\begin{table}[t]
\bgroup
\setlength\tabcolsep{2pt} 
\def\arraystretch{1.1} 

\caption{The complexity of the universality problems of one-counter nets in which weights are encoded in unary.} 	\label{tbl:Unary}
\begin{tabular}{c|c|c||c|c||c|c|}
\multirow{3}{*}{ \ML{Unary\\encoding}}&\multicolumn{2}{c||}{Universality}&\multicolumn{2}{c||}{Initial-Value Universality}&\multicolumn{2}{c|}{Bounded Universality}\\
\cline{2-7}
&\multicolumn{1}{c|}{\multirow{2}{*}{ \ML{Singleton\\Alphabet}}}& \multirow{2}{*}{ \ML{General\\Alphabet}}&\multirow{2}{*}{ \ML{Singleton\\Alphabet}} & \multirow{2}{*}{ \ML{General\\Alphabet}}& \multirow{2}{*}{ \ML{Singleton\\Alphabet}} & \multirow{2}{*}{ \ML{General\\Alphabet}}\\
&&&&&&\\
\hline\hline
Deterministic 
  &\Entry{ \L}  {\cref{lem:DOCN-singleton-complexities}}
  &\Entry{\NL-comp.\ }{\cref{thm:DOCN-complexities}}
  &\Entry{ \L}  {\cref{lem:DOCN-singleton-complexities}}
  &\Entry{\NL-comp.\ }{\cref{thm:DOCN-complexities}} 
  &\Entry{ \L}{\cref{lem:DOCN-singleton-complexities}} 
  &\Entry{\NL-comp.\ }{\cref{thm:DOCN-complexities}}\\
\hline
Unambiguous 
  &\Entry{ \NL}  {\cref{thm:UOCA-universality-unary-single}}
  &  \ML{\hspace{-.2cm}\myNC{2};\\[-8pt] \hspace{1.4cm}\RefSize{\cite{czerwiski:hal-02483495}} \\[-6pt]\hspace{-.4cm}\NL-hard}
  &\Entry{ \NL}{\cref{thm:SUOCN-iv-universality}}
  &\Entry{ \myNC{2}}{\cref{thm:SUOCN-iv-universality}}
  &\Entry{ \NL}{\cref{thm:UOCN-bu}} 
  &\Entry{ \myNC{2}}{\cref{thm:UOCN-bu}} 
  \\
\hline
\ML{Non-\\deterministic}
  &\Entry{\coNP-comp.}{\cref{thm:singleton alphabet unary coNP}}
  &\Entry{Ackermann}{\cite{HT2017}}
  &\Entry{\coNP-comp.\ }{ \cref{thm:Singleton-iv-universality}}
  &\Entry{Undecidable}{\cref{thm:Undecidable1}}
  &\Entry{\coNP-comp.}{\cref{thm:singleton alphabet bounded-univ complexity}}
  & \Entry{Undecidable}{\cref{thm:Undecidable2}}\\
\hline
\end{tabular}
\egroup
\end{table}

\begin{table}[t]
\bgroup
\setlength\tabcolsep{2pt}
\def\arraystretch{1.1} 

        \caption{The complexity of the bounded universality problems of one-counter nets in which weights are encoded in binary. } 	\label{tbl:Binary}
	\begin{tabular}{c|c|c||c|c||c|c|}
\multirow{3}{*}{ \ML{Binary\\encoding}}&\multicolumn{2}{c||}{Universality}&\multicolumn{2}{c||}{Initial-Value Universality}&\multicolumn{2}{c|}{Bounded Universality}\\
\cline{2-7}
&\multicolumn{1}{c|}{\multirow{2}{*}{ \ML{Singleton\\Alphabet}}}& \multirow{2}{*}{ \ML{General\\Alphabet}}&\multirow{2}{*}{ \ML{Singleton\\Alphabet}} & \multirow{2}{*}{ \ML{General\\Alphabet}}& \multirow{2}{*}{ \ML{Singleton\\Alphabet}} & \multirow{2}{*}{ \ML{General\\Alphabet}}\\
&&&&&&\\
\hline\hline
		Deterministic 
                  &\Entry{\myNC{2}}{\cref{lem:DOCN-singleton-complexities}}
                  &\Entry{ \NC}{\cref{thm:DOCN-complexities}}
                  &\Entry{ \myNC{2}}{\cref{lem:DOCN-singleton-complexities}}
                  &\Entry{ \myNC{2}}{\cref{thm:SUOCN-iv-universality}} 
                  &\Entry{ \myNC{2}}{\cref{lem:DOCN-singleton-complexities}}
                  &\Entry{ \myNC{ }}{\cref{thm:DOCN-complexities}}
		\\
		\hline
		Unambiguous 
                  &\Entry{\coNP-comp.}{\cref{thm:singleton alphabet binary Sigma2}}
	          &\ML{\hspace{-.4cm}PSPACE;\\[-8pt] \hspace{1.5cm}\RefSize{\cite{czerwiski:hal-02483495}}\hspace{-.15cm} \\[-6pt]\hspace{-.3cm}coNP-hard}
                  &\Entry{ \myNC{2}}{\cref{thm:SUOCN-iv-universality}}
                  &\Entry{ \myNC{2}}{\cref{thm:SUOCN-iv-universality}} 
                  &\Entry{$\coNP^\NP$}{\cref{thm:singleton alphabet bounded-univ complexity}}
                  &\Entry{ \PSPACE}{\cref{thm:UOCN-bu-bin}} \\
		\hline
		\ML{Non-\\deterministic}
                  &\Entry{$\coNP^\NP$}{\cref{thm:singleton alphabet binary Sigma2}}
                  &\Entry{Ackermann}{\cite{HT2017}}
                  &\Entry{\coNP-comp.\ }{ \cref{thm:Singleton-iv-universality}}
                  &\Entry{Undecidable}{\cref{thm:Undecidable1}}
                  &\Entry{$\coNP^\NP$}{\cref{thm:singleton alphabet bounded-univ complexity}}
                  &\Entry{Undecidable}{\cref{thm:Undecidable2}}\\
		\hline
	\end{tabular}
	\egroup
\end{table}

\subparagraph*{Related work.}
The undecidability of language universality for pushdown automata is textbook.
In his 1973 PhD thesis \cite{Val1973}, Valiant showed that the problem remains
undecidable for the strictly weaker model of one-counter automata (OCA, with
zero tests) by recognizing the complement of all accepting runs of a two-counter machine.
Language inclusion is undecidable for the further restricted model of OCNs \cite{HLMT2016}.
If one considers $\omega$-regular languages defined by OCNs with B\"uchi acceptance condition
then the resulting universality problem is undecidable \cite{BGHH2017}.

On the positive side, universality is decidable 
for vector addition systems \cite{JEM1999}
and Ackermann-complete for the special case of OCNs \cite{HT2017}.
One-counter systems have received some attention in regards to checking
bisimulation and simulation relations, which under-approximate language
equivalence (and inclusion, respectively) and are computationally simpler.
For OCAs/OCNs, bisimulation is \PSPACE-complete \cite{BGJ2010}, while weak
bisimulation is undecidable for OCNs \cite{Mayr2003}.
Both strong and weak simulation are \PSPACE-complete for OCNs, and checking if an
OCN simulates an OCA is decidable \cite{AAHMKT2014}.

Universality problems for OCNs over single-letter alphabets
are related to the termination problem for VASS,
which asks if there exists an infinite run.
Non-termination naturally corresponds to the property that $a^n \in \Lang{\initialstate,\vec{v_0}}$,
i.e., all finite words are accepted, assuming that all states are accepting.
Termination reduces to boundedness (finiteness of the reachability set) which is
\EXPSPACE-complete \cite{Rac1978,Dem2013} in general and
\PSPACE-complete for systems with fixed dimensions \cite{RY1985}. In contrast, the \emph{structural} termination problem (there exists no infinite run, regardless of the initial configuration) is equivalent to finding an executable cycle that is non-decreasing on all dimensions, and can be solved in polynomial time \cite{KS1988}.

Finally, the idea to existentially quantify over some initial resource is
commonplace in the formal verification literature.
Examples include unknown initial-credit problems
for energy games \cite{BFLMS2008,AAHMKT2014} and R-Automata \cite{AKY2008},
timed Petri nets \cite{AACMT2018}, and inclusion problems for weighted automata \cite{DDGRT10,ABK11}.

We defer most proofs to the Appendix.

\section{Preliminaries}
\label{sec:preliminaries}
\subparagraph*{One-Counter Nets.}
A \emph{one-counter net} (OCN)
is a finite directed graph where edges carry both an integer weight and a letter from a finite alphabet.
We write
$\sys{A}=\ocntuple$ 
for the net $\sys{A}$ where $\states$ is a finite set of \emph{states}, $\alphabet$ is a finite set of \emph{letters}, $\initialstate\in \states$ is an \emph{initial state},  $\transitions\subseteq \states\x\alphabet\x\Z\x\states$ is the \emph{transition} relation, and $\fstates\subseteq \states$ are the \emph{accepting} states.
%

For a transition $t=(s,a,e,s')\in\transitions$ we write $\effect{t}\eqdef e$
for its (counter) \emph{effect}, and write $\norm{\delta}$ for the largest absolute effect among all transitions.
By the \emph{underlying automaton} of an OCN
we mean the NFA 
obtained from the OCN by disregarding the transition effects.

A path in the OCN is a sequence $\pi=(s_1,a_1,e_1,s_2)(s_2,a_2,e_2,s_3)\dots (s_{k},a_k,e_k,s_{k+1})\in\transitions^*$. Such a path $\pi$ is a \emph{cycle} if $s_1=s_{k+1}$, 
and is a \emph{simple cycle} if no other cycle is a proper infix of it.
We say that the path above \emph{reads} word $a_1a_2\dots a_k\in\alphabet^*$ and is accepting if 
$s_{k+1}\in\fstates$.
Its $\effect{\pi}\eqdef \sum_{i=1}^k e_i$ is the sum of its transition effects .
Its \emph{height} is the maximal effect of any prefix and, similarly, its
\emph{depth} is the inverse of the minimal effect of any prefix.

\medskip
An OCN naturally induces an infinite-state labelled transition system in which each
\emph{configuration} is a pair $(s,c)\in \states\x\N$ comprising a state and a non-negative integer.
We 
call such a configuration \emph{final}, or \emph{accepting}, if $s\in F$.
Every letter $a\in\alphabet$ induces a step relation
$\step{a}~\subseteq (Q\x\N)^2$ between configurations where, for every two configurations $(s,c)$ and $(s',c')$,
$$
(s,c)\step{a}(s',c') \iff (s,a,d,s')\in\transitions
\quad\text{and }
c'=c+d. 
$$
A \emph{run} on a word $w=a_1a_2\ldots a_k\in\alphabet^*$
is a path in this induced infinite system; that is,
a sequence $\rho=(\initialstate,c_0),(s_1,c_1), (s_2,c_2),\ldots (s_k,c_k)$
such that $(s_{i-1},c_{i-1})\step{a_i}(s_i,c_i)$ holds for all $1\le i\le k$.
Naturally, a run uniquely describes a path in the underlying finite OCN.
Conversely, for every such path and initial counter value $c_0\in\N$, there is at most one
corresponding run: A path $\pi$ is \emph{executable from $c_0$} if
its depth is at most $c_0$ (that is, we do not allow the counter to become negative).
A run as above is called a (simple) 
\emph{cycle}
if its underlying path
is a (simple) cycle.
%
It is \emph{accepting} if it ends in an accepting configuration.
We call a run \emph{bounded} by $b\in\N$
if $c_i\le b$ for all $0\le i\le k$.

For any fixed initial configuration $(s,c)$, 
we define its \emph{language}
$\Lang[\?A]{s,c}\subseteq\alphabet^*$
to contain exactly all words on which an accepting run starting in $(s,c)$ exists. (We omit the subscript $\?A$ if the OCN is clear from context.)
Similarly, the \emph{$b$-bounded language} $\bLang{s,c}{b}$
is the set of those words on which there is a $b$-bounded run starting in $(s,c)$.


\medskip
The OCN is
\emph{deterministic} if for every pair $(s,a) \in Q\x\alphabet$ there
is at most one pair $(d,q)\in\N\x Q$ with $(s,a,d,s')\in\delta$.
A net together with an initial configuration $(\initialstate,c_0)$
is \emph{unambiguous} if for every word $w\in\alphabet^*$
there is at most one accepting run starting in $(\initialstate,c_0)$.

\subparagraph*{Two-Counter Machines.}
A two-counter machine (Minsky Machine) $\M$ is a sequence $(l_1,\ldots,l_n)$ of commands involving two counters $x$ and $y$. We refer to
$\set{1,\ldots,n}$ as the {\em locations} of the machine. There are five possible forms of commands:
\texttt{inc(c)}, \texttt{dec(c)}, \texttt{goto $l_i$}, \texttt{halt}, \texttt{if c=0 goto $l_i$ else goto $l_j$}, 
where $c\in \set{x,y}$ is a counter and $1\le i,j\le n$ are locations. 
The counters are initially set to $0$.
Since we can always check whether $c=0$ before a $\texttt{dec(c)}$ command, we assume that the machine never reaches $\texttt{dec(c)}$ with $c=0$. That is, the counters never have negative values. 

\section{Undecidability}
\label{sec:nondet}
We show that both initial-value universality and bounded universality are undecidable by reduction from the undecidable halting problem of two-counter machines (2CM) \cite{Min67}.

The idea underlying both reductions is that
the initial counter value, or the bound on the allowed counter, prescribes a bound on the number of steps until the OCN must make a decision weather the input word, which encodes a prefix of the run of the 2CM, 
either halts or cheats.
After this decision the OCN is reset and continues to read the remaining word within an adjusted bound.
If the decision was correct then the bound remains the same and otherwise, it is strictly reduced.
The existence of a halting run of the 2CM 
now implies that its length corresponds to a sufficient initial bound for this simulating OCN to be universal.
Conversely, if the run of the machine does not halt
then for every bound $n$, 
there exists a non-cheating, and non-terminating prefix of length $n$.
Repeating this prefix $n$ times
witnesses non-universality for the simulating OCN with initial counter $n$.

\subsection{Initial-Value Universality}\label{sec:NondetInitial}

\begin{figure}
	\centering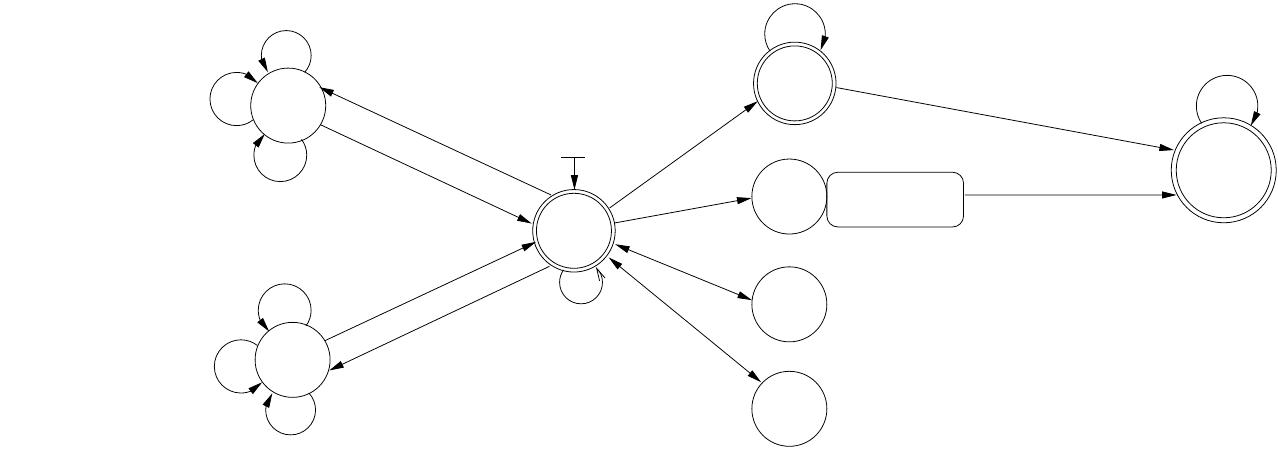 \caption{The one-counter net $\A$ from the proof of \cref{thm:Undecidable1}.}\label{fig:A}
\end{figure}

Given a two-counter machine $\M$, we construct a one-counter net $\A$ as follows (see Figure~\ref{fig:A}).
%
Intuitively, an input word $w$ to $\A$ is a sequence of segments separated by $\#$, where each segment is a sequence of commands from $\M$.
Accordingly, the alphabet of $\A$ consists of $\#$ and all possible commands of $\M$.

We build $\A$ to accept $w$, once starting with a big enough initial counter value, if one of the following conditions holds: i) one of $w$'s segments is shorter than the length of the (legal halting) run of $\M$; or ii) one of $w$'s segments does not respect the control structure underlying $\M$, which is called a ``non-counting cheat'' here; or iii) all of $w$'s segments do not describe a prefix of the run of $\M$, making ``counting cheats''.
The OCN reads every segment in between two $\#$'s starting in, and returning to, a central state $q_0$.

Non-counting cheats are easy to verify---for every line $l$ of $\M$, there is a corresponding state $q$ in $\A$, and when $\A$ is at state $q$ and reads a letter $a$, $\A$ checks if $a$ matches the command in $l$. For example, if $l=$\Command{goto i} and $a=$ \Command{inc~x}, the transition from $q$ goes to a forever accepting state ($heaven$), and if $a=$\Command{goto i}, it goes to the state of $\A$ that corresponds to the line $l_i$. This is the  ``command-checker gadget'' of $\A$.

Counting cheats are more challenging to verify, as OCNs cannot branch according to a counter value. We consider separately ``positive cheats'' and ``negative cheats''. The former stands for the case that the input letter is \Command{x=0~then~goto}  (or  \Command{y=0~then~goto}) while the value of $x$ (or $y$) in the legal run of $\M$ should be positive. The latter stands for the case that the input letter is \Command{x>0~then~goto}  (or  \Command{y>0~then~goto}) while the value of $x$ (or $y$) in the legal run of $\M$ should be $0$.

Positive cheats can be verified by directly simulating the respective counter of $\M$ using the counter in $\A$
(states $q_3$ and $q_5$ in \cref{fig:A}). Once the cheat occurs, $\A$ can return to $q_0$ with a penalty of $-1$, and since the counter in $\M$ is positive, we are guaranteed that the counter in $\A$ did not decrease since leaving $q_0$, allowing $\A$ to continue the run.

For verifying a negative cheat, we simulate the counting of $\M$ by an ``opposite-counting'' in $\A$
(states $q_4$ and $q_6$ in \cref{fig:A}), whereby an increment of the counter in $\M$ results in a decrement of the counter in $\A$, and vice versa---once the cheat occurs, $\A$ can return to $q_0$ with no penalty, and since the counter in $\M$ is $0$, we are guaranteed that the counter in $\A$ did not  decrease since leaving $q_0$, allowing $\A$ to continue the run.

\medskip
\noindent
Formally, we construct $\A$ from $\M$ as follows.
\begin{itemize}
	\item The alphabet $\Sigma$ of $\A$ consists of $\#$ and the descriptive commands for the counter machine $\M$ : \Command{inc~x}, \Command{inc~y}, \Command{dec~x}, \Command{dec~y}, \Command{halt}, and for every line $i$ of $\M$, the commands \Command{goto i}, \Command{x=0~then~goto i}, \Command{y=0~then~goto i}, \Command{x>0~then~goto i}, and \Command{y>0~then~goto i}.
	\item The initial state $q_0$ is accepting, it has a self transition over $\Sigma\setminus\{\#\}$ and nondeterministic transitions to the states $q_1\ldots q_6$ over $\#$, all with weight $0$.
	\item There is a $heaven$ state, which is accepting, and has a self loop over $\Sigma$ with weight $0$.
	\item The state $q_1$ is accepting and intuitively allows to accept short segments between consecutive $\#$'s: It has a self transition over $\Sigma\setminus\{\#\}$ and a transition to $heaven$ over $\#$, all with weight $-1$.
	\item The state $q_2$ starts the command-checker gadget, which looks for a non-counting violation of $\M$'s commands (which is a simple regular check). Once reaching a violation it goes to $heaven$. All of its transitions are with weight $0$. If it does not find a violation, it cannot continue the run.
	\item The state $q_3$ is a positive-cheat checker for $\M$'s counter $x$. 
	It has a self loop over \Command{inc~x} with weight $+1$ and over \Command{dec~x} with weight $-1$. Over \Command{x=0~then~goto} it can nondeterministically choose between a self loop with weight $0$ and a transition to $q_0$ with weight $-1$. Over the rest of the alphabet lettres, except for \Command{halt} and $\#$, it has a self loop with weight $0$. (Over \Command{halt} and $\#$ it cannot continue the run.)
	\item The state $q_4$ is a negative-cheat checker for $\M$'s counter $x$. 
	It has a self loop over \Command{inc~x} with weight $-1$ and over \Command{dec~x} with weight $+1$. Over \Command{x>0~then~goto} it can nondeterministically choose between a self loop with weight $0$ and a transition to $q_0$ with weight $0$. Over the rest of the alphabet lettres, except for \Command{halt} and $\#$, it has a self loop with weight $0$.
	\item The states $q_5$ and $q_6$ provide positive-cheat checker and negative-cheat checker for $\M$'s counter $y$, respectively, analogously to states $q_3$ and $q_4$.
\end{itemize}

%
\begin{theorem}
\label{thm:Undecidable1}
	The initial-value universality problem for one-counter nets is undecidable.
\end{theorem}
\begin{proof}

	We show that a given two-counter machine $\M$ halts if and only if the corresponding one-counter net $\A$, as constructed in \cref{sec:NondetInitial}, is initial-value universal.

	\subproof $\Rightarrow:$ \emph{If $\M$ halts}, its (legal) run has some length $n-1$. We claim that $\A$ is universal with the initial value $n$.
	
	Consider some word $w$ over the alphabet of $\A$. We shall describe an accepting run $\rho$ of $\A$ on $w$.
	Until the first occurrence of $\#$, the run $\rho$ is deterministically in $q_0$, which is accepting.
	We show that for every segment between two consecutive $\#$'s, as well as the segment after the last $\#$, the run $\rho$ may either reach $heaven$ or reach $q_0$ with counter value at least $n$ (and remains there until the next $\#$ or the end of the word), from which it follows that $\rho$ is accepting.
	
	If the segment is shorter than $n$, $q_0$ can choose to go to $q_1$ over $\#$, and from there it will reach heaven.
	If the segment is longer than $n$, it cannot describe the legal run of $\M$. Then, it must cheat within up to $n$ steps. We show that each of the 5 possible cheats fulfills the claim.
	\begin{description}
		\item [1.] If it makes a non-counting cheat, $q_0$ will go to $q_2$ over $\#$, and will reach $heaven$. (This is also the case if it has additional letters different from $\#$ after the \Command{halt} letter.)
		\item [2.] If it makes a positive cheat on $x$, $q_0$ will go to $q_3$ upon reading the next $\#$. When the cheat occurs, the value of $x$ is positive, while reading the letter \Command{x=0~then~goto}. Notice that the value of $\A$'s counter is accordingly bigger than its value when entering $q_3$ (and by the inductive assumption bigger than $n$). Then, $q_3$ goes to $q_0$ with weight $-1$, guaranteeing that $\A$'s counter value is at least $n$.
		Notice that the counter value cannot go below $n$ at any point, since $\M$ cannot make the value of $x$ negative without a counting cheat. (We equipped $\M$ with a counter check before every decrement.)
		\item [3.] If it makes a negative cheat on $x$, $q_0$ will go to $q_4$. Then, when the cheat occurs, the value of $x$ is $0$, while there is the letter \Command{x>0~then~goto}. Notice that the value of $\A$'s counter is accordingly exactly its value when entering $q_3$ (and by the inductive assumption at least $n$). Then, $q_4$ goes to $q_0$ with weight $0$, guaranteeing that $\A$'s counter value is at least $n$.
		Notice that the counter might go below $n$ between getting to $q_4$ and returning to $q_0$. Yet, since the violation must occur within up to $n$ steps, and the value of the counter when entering $q_4$ is at least $n$, we are guaranteed to be able to properly continue with the run, as the counter need not go below $0$.
		\item [4-5.] Analogously, if it makes a positive or negative cheat over $y$, the choice of $q_0$ will be $q_5$ or $q_6$, respectively.
	\end{description}
	
	\subproof $\Leftarrow:$ \emph{If $\M$ does not halt}, for every positive integer $n$, we build the word $w_n$ and show that it is not accepted by $\A$ with an initial counter value $n$.
	
	The word $w_n$ consists of $n+1$ segments between $\#$'s, where each segment is the prefix of length $n+1$ of the (legal) run of $\M$.
	Consider the possible runs of $\A$ on $w_n$. It cannot go from $q_0$ to $q_1$, because it will stop after $n$ steps. It also cannot go to $q_2$, because there is no cheating.
	We show that if it goes to $q_3..q_6$, it must return to $q_0$ before the next $\#$, while decreasing the value of $\A$'s counter, which can be done only $n$ times until the run stops.
	
	If it goes to $q_3$, it must return to $q_0$ upon some \Command{x=0~then~goto}, as it cannot continue the run on $\#$. Yet, as there is no cheating, it returns to $q_0$ when $x=0$, which implies that $\A$'s counter has the same value as when entering $q_3$, and due to the $-1$ weight of the transition to $q_0$, it returns to $q_0$ while decreasing the value of $\A$'s counter by $1$. An analogous argument follows if it goes to $q_5$.
	
	If it goes to $q_4$, it must return to $q_0$ upon some \Command{x>0~then~goto}, as it cannot continue the run on $\#$. Yet, as there is no cheating, it returns to $q_0$ while the value of $x$ is indeed strictly positive, which implies that the value of $\A$'s counter is smaller than the value it had when entering $q_4$, and therefore due to the $0$-weight transition to $q_0$, it returns to $q_0$ with a smaller value of $\A$'s counter. An analogous argument follows if it goes to $q_6$.
\end{proof}

\subsection{Bounded Universality}\label{sec:NondetBounded}

We show that the problem is undecidable by making some changes to the undecidability proof of the initial-value universality problem.

Given a two-counter machine $\M$, we construct a one-counter net $\A'$ that is similar to $\A$, as constructed above, except for the following changes (see Figure~\ref{fig:A'}):
\begin{itemize}
	\item There is an additional state $q'_0$ that is accepting, it is the new initial state, and it has a nondeterministic choice over $\Sigma$ of either taking a self loop with weight $+1$ or going to $q_0$ with weight $0$.
	\item The state $q_0$ is no longer initial, and it has an additional transition over $\#$ to a new state $q_7$ with weight $0$.
	\item The state $q_7$ is accepting, and it has nondeterministic choice over $\Sigma$ of either taking a self loop with weight $-1$ or going to $q_0$ with weight $-1$.
\end{itemize}

\begin{figure}
	\centering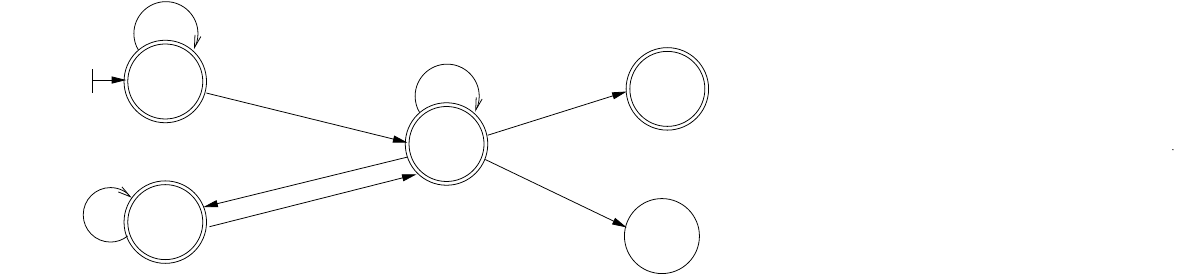 \caption{The one-counter net $\A'$ from the proof of \cref{thm:Undecidable2}.}\label{fig:A'}
\end{figure}

\noindent
Now $\M$ halts if and only if $\A'$ is bounded universal for an initial counter value $0$.
A detailed proof can be found in \cref{apx:thm:Undecidable1}.

\begin{restatable}{theorem}{undecBU}
\label{thm:Undecidable2}
	The bounded universality problem for one-counter nets is undecidable.
\end{restatable}

\section{Singleton Alphabet}
\label{sec:unary}
In this section we study universality problems on OCN over singleton alphabets. 
The universality problem for NFA over singleton alphabets is already \coNP-hard \cite{MS1973},
a lower bound which trivially carries over to all problems considered here\footnote{	
	The proof in \cite[Theorem~6.1]{MS1973}	
	in fact shows \NP-completeness of the	
	problem of whether two regular expressions over $\{0\}$ define different languages. Hardness is shown by reduction from Boolean satisfiability to non-universality of expressions using prime-cycles, and it is straightforward to rephrase it in terms of DFAs.	
}.

For simplicity, we identify languages $L\subseteq\{a\}^*$ with their Parikh image, so that the universality problems	
ask if the (bounded) language of a given OCN equals $\N$.	
Throughout this section, fix an OCN $\?A=\ocntuple$.

We start by sketching our approach. Observe that the language of an OCN is not universal iff the OCN does not accept some word $w$. To show that such $w$ exists, we distinguish between two cases: either $w$ is ``relatively short'', in which case we use a guess-and-check approach to find it, or it is long, in which case we deduce its existence by analyzing some cyclic behaviour of the OCN. The details of both the guess-and-check elements and the cyclic behaviour depend on the encoding of the weights and the variant of universality.

\subsection{Universality}
\label{subsec: singleton universality}
We start by describing a procedure to decide the ordinary universality problem for OCN over singleton alphabets -- with fixed initial configuration and no bounds on the counter.

Consider a cycle $\gamma=s_1,s_2,\ldots,s_k$ (with $s_1=s_k$). Recall that $\effect{\gamma}$ is the sum of weights along $\gamma$ and $\depth{\gamma}$ is the inverse of the lowest effect along the prefixes of $\gamma$. We call $1\le d\le k$ a \emph{nadir} of $\gamma$ if it is the index of a prefix that attains the depth of $\gamma$. That is, $\effect{s_1,\ldots, s_d}=-\depth{\gamma}$. 
We say that $\gamma$ is \emph{positive} if $\effect{\gamma}$ is positive (and similarly for negative, non-negative, zero, etc.).
We call $\gamma$ \emph{good} if it a simple, non-negative cycle, and $\depth{\gamma}=0$. 

\begin{observation}
	\label{obs:shifted cycle}
	If $\gamma$ is non-negative and it has a nadir $d$, then the \emph{shifted cycle} $\gamma^{\leftarrow d}\eqdef s_d\, s_{d+1},\cdots , s_k,s_2,\cdots ,s_d$ is good.
	Similarly, if $\gamma$ is negative, then $\effect{\gamma^{\leftarrow d}}=-\depth{\gamma^{\leftarrow d}}$.
\end{observation}

For a state $r\in\states$ and an initial configuration $\initialstate,c_0$,	
let $\LangVia{\initialstate,c_0}{r} \subseteq \Lang{\initialstate,c_0}$ be the language of words accepted by a run that visits $r$.

The first tool we use in studying the universality problem is a canonical form for accepting runs, 
akin to \emph{linear path schemes} of \cite{LS2004,BFGHM15}.
\begin{definition}[Linear Forms]
	\label{def:linear path scheme}
	A path $\pi$ is in \emph{linear form} if there exist simple cycles $\gamma_1,\ldots,\gamma_k$ and paths $\tau_0,\ldots,\tau_{k}$ such that  $\pi=\tau_0\gamma_1^{e_1}\tau_1\cdots \tau_{k-1}\gamma_k^{e_k}\tau_k$ for some numbers $e_1,\ldots,e_k\in\N$, and such that  every non-negative cycle $\gamma_i$, is taken from a nadir, and so is executable with any counter value.
	
	We call $e_i$ the \emph{exponent} of $\gamma_i$, and we refer to $\tau_0\gamma_1\tau_1\ldots\gamma_k\tau_k$ as the \emph{underlying path} of $\pi$. The \emph{length} of the linear form is the length of the underlying path.
\end{definition}

A linear form is described by the components above, where the exponents are given in binary.
In the following, we show that every path can be transformed to a path in linear form with a small description size.
\begin{restatable}{lemma}{unaryLPF}
	\label{lem:linear form of paths}
	Let $\pi$ be an executable path of length $n$ from $(p,c)$ to $(q,c')$.
        Then there exists an executable path $\pi'$ of length $n$ in linear form whose length is 
        at most $2|\states|^2$, from $(p,c)$ to $(q,c'')$ with $c''\ge c'$.
\end{restatable}
\begin{proof}[Proof Sketch:]
	$\pi'$ is obtained from $\pi$ in two steps, namely rearranging simple cycles, and then choosing a small set of ``representative'' simple cycles to replace others. The crux of the proof is the first step, where instead of simply moving a cycle, we also shift it so that it is taken from its nadir. Then, for every set of simple cycles of the same length and on the same state, we take the one with maximal effect as a representative.
\end{proof}

We now turn to identify states that have a special significance in analyzing universality.
\begin{definition}
	\label{def:pump states}
	Let $\Pump\subseteq\states$ be the set of states that admit good cycles. For each such state $r$ fix a shortest good cycle $\gamma_r$.
\end{definition}
Intuitively, a state $r$ is in \Pump if it has a cycle that can be taken with any counter value, any number of times. That is, it can be used to ``pump'' the length of the word.
Another important property is that if a path never visits a state in \Pump
then \emph{all} its simple cycles must be negative.
Indeed, any non-negative cycle must contain a non-negative simple cycle
and any state at a nadir of such cycle must be in \Pump.

If however, a state in \Pump occurs along an accepting run, we can accept the same word using a
run in a short linear form, as we now show.
\begin{restatable}{lemma}{unaryLPFpump}
	\label{lem:linear form with pump state}
	There exists a bound $\bound{1}\in\poly(|\states|,\norm{\transitions})$ such that, for every $n\in\N$,
	if $n$ is accepted by a run that visits a state $r\in \Pump$, then $n$ has an accepting run of the form $\eta_1\gamma_r^t\eta_2$ for paths $\eta_1,\eta_2$ of length at most $\bound{1}$.
\end{restatable}
\begin{proof}[Proof Sketch:]
	Using \cref{lem:linear form of paths}, we split an accepting run on $n$ that visits $r$ to the form $\pi_1,r,\pi_2$ where $\pi_1$ and $\pi_2$ are in linear form. Then, we successively shorten $\pi_1$ and $\pi_2$ by eliminating simple cycles along them, and instead pumping the non-negative cycle $\gamma_r$. Some careful accounting is needed so that the length of the path is maintained, and so that it remains executable.
\end{proof}

We now characterize the regular language $\LangVia{\initialstate,c_0}{r}$ using a DFA of bounded size. 

\begin{restatable}{lemma}{unaryPumpDFA}
	\label{lem:pump state DFA}
	There exists a bound $\bound{2}\in\poly(\norm{\transitions}\cdot\card{\states})$
	such that, 
	for every $r\in \Pump$, there exists a DFA 
	that accepts $\LangVia{\initialstate,c_0}{r}$ and is of size at most
	$\bound{2}$.
\end{restatable}

Define $\PS\eqdef\bigcup_{r\in \Pump} \LangVia{\initialstate,c_0}{r}$. 
Notice that $\PS\subseteq \Lang{\initialstate,c_0}$
and that $\Lang{\initialstate,c_0}\setminus \PS$ must be finite. 
Indeed, if $w\in \Lang{\initialstate,c_0}\setminus \PS$ then it can only be accepted by runs with \emph{only} negative cycles, of which there are finitely many. In particular, if $\N\setminus \PS$ is infinite, then $\Lang{\initialstate,c_0}\neq \N$. 

Using the bounds from \cref{lem:pump state DFA}, we have the following.
\begin{restatable}{lemma}{unaryWitnessBound}
	\label{lem:singleton alphabet witness bound}
	There exists $\bound{3}\in\poly(\norm{\transitions},\card{\states})$ such that $\Lang{\initialstate,c_0}\neq \N$ if, and only if, there exists $n\in \N$ such that either $n<\bound{2}$ and $n\notin \Lang{\initialstate,c_0}$, or $\bound{3}^{|\states|}\le n\le 2\bound{3}^{|\states|}$ and $n\notin \PS$.
\end{restatable}

\Cref{lem:singleton alphabet witness bound} suggests the following algorithmic scheme for deciding non-universality: 
non-deterministically either (1) guess $n<\bound{3}$, and check that $n\notin \Lang{\initialstate,c_0}$, or (2) guess $\bound{3}^{|\states|}\le n\le 2\bound{3}^{|\states|}$ and check that $n\notin \LangVia{\initialstate,c_0}{r}$ for all $r\in \Pump$, which implies that $n\notin\?P$.

Note that even if the transitions are encoded in unary, $n$ still needs to be guessed in binary for part (2) (and also for part (1) if the encoding is binary). 
The complexity of the checks involved in both parts of the algorithm depend on the encoding of the transitions, and are handled separately in the following.

\subparagraph*{Unary Encoding.} 
If the transitions are encoded in unary, then $\bound{3}$ is polynomial in the size of the OCN. Consequently, we can check for $n<\bound{3}$ whether $n\in \Lang{\initialstate,c_0}$ by simulating the OCN for $n$ steps, while keeping track of the maximal run to each state. Indeed, due to the monotonicity of executability of OCN paths
it suffices to remember, for each state $s$,
the maximal possible counter-value $c$ so that $(s,c)$ is reachable via the current prefix,
which must be a number $\le c_0 + n\cdot \norm{\transitions}$ 
or $-\infty$ (to represent that no configuration $(s,c)$ can be reached).

Next, in order to check whether $n\notin \LangVia{\initialstate,c_0}{r}$ for all $r\in \Pump$ for $\bound{3}^{|\states|}\le n\le 2\bound{3}^{|\states|}$ written in binary, we notice that since $\bound{3}$ is polynomial in the description of the OCN, then the size of each DFA for $\LangVia{\initialstate,c_0}{r}$ constructed as per \cref{lem:pump state DFA} is polynomial in the OCN. Since the proof in \cref{lem:pump state DFA} is constructive, we can obtain an explicit representation of these DFAs. 
Finally, given a DFA (or indeed, and NFA) over a singleton alphabet and $n$ written in binary, we can check whether $n$ is accepted in time $O(\log n)$ by repeated squaring of the transition matrix for the DFA \cite{MS1973}. 
We conclude with the following.
\begin{theorem}
	\label{thm:singleton alphabet unary coNP}
	The universality problem for singleton-alphabet one-counter nets with transitions encoded in unary is in \coNP, and is thus \coNP-complete.
\end{theorem}

\subparagraph*{Binary Encoding.}
When the transitions are encoded in binary, $\bound{3}$ is potentially exponential in the encoding of the OCN. Thus, naively adapting the methods taken in the unary case (with basic optimization) will lead to a $\PSPACE$ algorithm for universality (using Savitch's Theorem). As we now show, by taking a different approach, we can obtain an upper bound of $\SigmaTwo$, placing the problem in the second level of the polynomial hierarchy.

In order to obtain this bound, we essentially show that given $n$ encoded in binary, checking whether $n$ is accepted by the OCN can be done in \NP. This is based on the linear form of \cref{lem:linear form of paths}.
\begin{restatable}{lemma}{unaryCheckLPS}
	\label{lem:checking linear path scheme}
	Let $\pi=\tau_0\gamma_1^{e_1}\tau_1\cdots \tau_{k-1}\gamma_k^{e_k}\tau_k$ be a run in linear form, then we can check whether $\pi$ is executable from counter value $c$ in time polynomial in the description of $\pi$.
\end{restatable}

\Cref{lem:checking linear path scheme} shows that, given $n$ in binary, we can check whether $n\in \Lang{\initialstate,c_0}$ in $\NP$. Indeed, we guess the structure of an accepting run in linear form (including the exponents of the cycles), and check in polynomial time whether this run is executable, and whether it is accepting. 

In order to complete our algorithmic scheme for universality, it remains to show how we can check in $\NP$, given $n$ in binary, whether $n\notin \LangVia{\initialstate,c_0}{r}$ for every $r$. In contrast to the case of unary encoding, this is fairly simple.

Given $r$, we can construct an OCN $\?A^r$ such that $\Lang[\?A^r]{\initialstate,c_0}=\LangVia[\?A]{\initialstate,c_0}{r}$ by taking two copies of $\?A$, and allowing a transition to the second copy only once $r$ is reached. The accepting states are then those of the second copy.
Thus, checking whether $n\notin \LangVia{\initialstate,c_0}{r}$ amounts to checking whether $n\notin \Lang[\?A^r]{\initialstate,c_0}$. We can now complete the algorithmic scheme.
\begin{restatable}{theorem}{unaryUnivBinary}
	\label{thm:singleton alphabet binary Sigma2}
	The universality problem for singleton-alphabet one-counter nets with transitions encoded in binary is in $\SigmaTwo$.
\end{restatable}

\subsection{Initial-Value Universality}
The characterization of universality given in \cref{lem:singleton alphabet witness bound} can be simplified in the case of initial-value universality, in the sense that the freedom in choosing an initial value allows us to work with the underlying automaton of the OCN, disregarding the transition effects. This also allows us to obtain the same complexity results under unary and binary encodings.

Recall that $\Pump$ is the set of states that admit good cycles (see \cref{def:pump states}).
Let $\?N$ be the underlying NFA of $\?A$. For a state $r\in \Pump$, define $\LangVia[\?N]{\initialstate}{r}$ to be the set of words accepted by $\?N$ via a run that visits $r$. Overloading the notation of \cref{subsec: singleton universality}, we define $\PS\eqdef\bigcup_{r\in \Pump}\LangVia[\?N]{\initialstate}{r}$.
\begin{restatable}{lemma}{unaryIVchar}
	\label{lem:singleton alphabet init-univ characterization}
	There exists $c_0$ such that $\Lang[\?A]{\initialstate,c_0}=\N$ iff $\Lang[\?N]{\initialstate}=\N$ and $\N\setminus \PS$ is finite.
    \end{restatable}

Following similar arguments to those in \cref{lem:linear form with pump state,lem:pump state DFA}, and using the fact that we work with the underlying NFA, we can show the following.
\begin{restatable}{lemma}{unaryIVdfa}
	\label{lem:singleton alphabet init-univ DFA}
	There exists a bound $\bound{4}\in \poly(|\states|)$ such that,
	for every $r\in \Pump$ there exists a DFA that accepts
	$\LangVia{\initialstate}{r}$ and which is of size at most $\bound{4}$.
\end{restatable}

We can now solve the initial-value universality problem.
\begin{theorem}
	\label{thm:Singleton-iv-universality}
	The initial-value universality problem for one-counter nets (in unary or binary encoding) is \coNP-complete.
\end{theorem}
\begin{proof}
	First, observe that the problem is \coNP-hard by reduction from the universality problem for NFAs. We now turn to show the upper bound.
	
	By \cref{lem:singleton alphabet init-univ characterization}, it is enough to decide whether $\Lang[\?N]{\initialstate}=\N$ and $\N\setminus \PS$ is finite. Checking whether $\Lang[\?N]{\initialstate}=\N$,
	i.e., deciding the universality problem for NFA over a single-letter alphabet, can be done in \coNP~\cite{MS1973}. 
	
	By \cref{lem:singleton alphabet init-univ DFA},
	there exists a DFA $\?D$
	for $\N\setminus \PS$ of size
	at most $M=\bound{4}^{\card{\states}}$, by taking the intersection of the respective DFAs over every 
	$r\in\Pump$. Thus, $\N\setminus \PS$ is infinite iff $\?D$ accepts a word of length $M< n\le 2M$ (as such a word induces infinitely many other words). Thus, we can decide in $\NP$ whether $\N\setminus \PS$ is infinite, by guessing $M<n\le 2M$, and checking that it is in $\LangVia{\initialstate}{r}$ for every $r\in \Pump$
	(using repeated squaring on the respective DFAs).
	
	We conclude that both checking whether $\Lang[\?N]{\initialstate}=\N$ and whether $\N\setminus \PS$ is finite can be done in $\coNP$, and so the initial value universality problem is also in $\coNP$.
\end{proof}

\subsection{Bounded Universality}
For bounded universality, the states in \Pump are not restrictive enough: in order to keep the counter bounded, a state must admit a $0$-effect cycle. However, these cycles need not be simple. Thus, we need to adjust our definitions somewhat. Fortunately, however, once the correct definitions are in place, most of the proofs carry out similarly to those of \cref{subsec: singleton universality}.

\begin{definition}
	\label{def:stable states}
	A state $q\in \states$ is \emph{stable} if either:
	\begin{enumerate}
		\item it is at the nadir of a simple positive cycle, and admits a negative cycle, or
		\item it is at the nadir of a simple zero cycle.
	\end{enumerate}
	We denote by $\Stable$ the set of stable states.
\end{definition}
Identifying stable states can be done in polynomial time
(see e.g.~\cref{lem:DOCN-conditions}). The motivation behind this definition is to identify states
that admit a zero-effect (not necessarily simple) cycle.
\begin{restatable}{lemma}{unaryBUstable}
	\label{lem:singleton alphabet bounded-univ stable zero cycles}
	There exists a bound $\bound{5}\in \poly(|\states|,\norm{\transitions})$ such that,
	every stable state $q$ admits a zero cycle of length and depth at most $\bound{5}$.
\end{restatable}
By \cref{lem:singleton alphabet bounded-univ stable zero cycles}
we can fix, for each $q\in \Stable$, some zero-cycle $\zeta_q$ with effect and depth bounded by $\bound{5}$.
Recall that $\LangVia{\initialstate,c_0}{r}$ is the set of words that are accepted with a path that passes through $r$. 
Let $\?S\eqdef\bigcup_{r\in \Stable}\LangVia{\initialstate,c_0}{r}$.  We prove an analogue of \cref{lem:linear form with pump state}.
\begin{restatable}{lemma}{unaryBUlinearForm}
	\label{lem:singleton alphabet bounded-univ linear form stable state}
	There exists a bound $\bound{6}\in\poly(|Q|,\norm{\transitions})$ such that
	every $n\in \LangVia{\initialstate,c_0}{r}$ has an accepting run of the form $\eta_1\zeta_r^t\eta_2$ for paths $\eta_1,\eta_2$ of length at most $\bound{6}$.
\end{restatable}
\begin{proof}
	The proof follows \emph{mutatis-mutandis} that of \cref{lem:linear form with pump state}, with one important difference: before replacing cycles with iterations of the zero cycle $\zeta_r$, we replace a bounded number of cycles with the positive cycle on $r$, on which $r$ is at a nadir,\footnote{That is, unless $r$ is the nadir of a zero cycle, in which case the proof requires no changes.} so that the counter value goes above $\depth{\zeta_r}$, enabling us to take $\zeta_r$ arbitrarily many times.
	Note that this lengthens the prefix $\eta_1$ at most polynomially in $(|\states|\cdot \norm{\transitions})$.
\end{proof}

\Cref{lem:singleton alphabet bounded-univ linear form stable state} 
implies that every word $n\in\?S$ 
can be accepted by a run whose counter values are bounded
because there must by an accepting run that, except for some bounded prefix and suffix,
only iterates some zero-cycle $\zeta_r$.
More precisely, we have the following.
\begin{theorem}
	\label{cor:singleton alphabet bounded-univ bounded langvia}
	There exists $\bound{6}\in \poly(|\states|,\norm{\transitions})$ such that every word $n\in \?S$ is accepted by a run whose counter value remains below $2\bound{6}+c_0$.
\end{theorem}
In addition, \cref{lem:singleton alphabet bounded-univ linear form stable state} immediately gives us (with an identical proof) an analogue of \cref{lem:pump state DFA}.
\begin{restatable}{lemma}{unaryBUdfa}
	\label{lem:singleton alphabet bounded-univ DFA}
	There exists a bound $\bound{7}\in\poly(|Q|,\norm{\transitions})$ such that,
	for every $r\in \Stable$ there exists a DFA that accepts 
	$\LangVia{\initialstate,c_0}{r}$ and is of size at most $\bound{7}$.
\end{restatable}
We can now characterize bounded universality in terms of $\?S$, the set of stable states.
\begin{restatable}{lemma}{unaryBUchar}
	\label{lem:singleton alphabet bounded-univ characterization}
	$\Lang{\initialstate,c_0}$ is bounded-universal if, and only if, 
	the underlying automaton $\?N$ is universal ($\Lang[\?N]{\initialstate}=\N$) and $\N\setminus \?S$ is finite.
\end{restatable}

Finally, checking whether $\N\setminus \?S$ is finite can be done similarly to \cref{subsec: singleton universality} (and the complexity depends on the transition encoding),
by checking that a candidate word $n$ of bounded length is not in $\LangVia{\initialstate,c_0}{r}$
for all stable states $r$. We conclude with the following.
\begin{theorem}
	\label{thm:singleton alphabet bounded-univ complexity}
	Bounded universality of one-counter nets is $\coNP$-complete assuming unary encoding, and in $\SigmaTwo$ assuming binary encoding.
\end{theorem}

\section{Deterministic Systems}
\label{sec:deterministic}
We turn to deterministic one-counter nets (DOCNs) for which the underlying finite automaton is a DFA.
We assume without loss of generality that the 
graphs underlying the DOCNs are connected, i.e., that all states are reachable from the initial state. 

For such systems, (bounded) universality problems can be decided by checking a suitable combination of simple conditions on cycles and short words.
In order to prevent tedious repetition, we list these conditions first and prove (in \cref{apx:lem:DOCN-conditions}) upper bounds for checking each of them (\cref{lem:DOCN-conditions}). We then show which combination allows to solve each decision problem 
(\cref{lem:DOCN-universalities-char}).

All mentioned upper bounds follow either easily from first principles,
or from the result that the state reachability problem (a.k.a., coverability) for OCN
is in \NC\ \cite[Theorem 15]{almagor2019coverability}.
We will also use the following fact, which follows from~\cite{IntroductionToCircuitComplexity} (see \ref{apx:lem:binary-addition}).
\begin{restatable}{lemma}{detBinaryAddition}
	\label{lem:binary-addition}
	\label{lem:addition of n numbers in binary}
	Given a set $S=\{\alpha_1, \alpha_2\ldots \alpha_n\}$ of integers written in binary, the question whether the sum of all elements in $S$ is non-negative is in \myNC{2}. 
\end{restatable}
%
\begin{restatable}[Basic Conditions]{lemma}{detConditions}
	\label{lem:DOCN-conditions}
	Consider the following conditions
	on a deterministic one-counter net
	$\sys{A}=\ocntuple$, 
	initial value $c_0\in\N$, and bound $b\in\N$.
	\begin{description}
		\item[\CondAutUniversal]
		The underlying automaton is universal.
		\item[\CondAllShortAccepting]
		Every word $w$ of length $\len{w}\le \card{\states}$ is in $\Lang{\initialstate,c_0}$
		\item[\CondAllShortBAccepting]
		Every word $w$ of length $\len{w}\le \card{\states}$ is in $\bLang{\initialstate,c_0}{b}$
		\item[\CondNNSimpleCycles]
		All simple cycles have non-negative effect.
		\item[\CondAllCyclesZero]
		All simple cycles have $0$-effect.
	\end{description}
	Condition \refCond{\CondAutUniversal} can be checked in
	non-deterministic logspace (\NL), 
	independently of the encoding of numbers.
	All other conditions can be verified in \NL\ assuming unary encoding,
	and in \NC\ (conditions \refCond{\CondNNSimpleCycles} and \refCond{\CondAllCyclesZero} even in \myNC{2}) assuming binary encoding.
\end{restatable}

\begin{restatable}{lemma}{detUs}
\label{lem:DOCN-universalities-char}
Consider a deterministic one-counter net with initial state $\initialstate$.
\begin{enumerate}
    \item For any $c_0\in\N$,
	the language $\Lang{\initialstate,c_0}$ is universal if, and only if,
	all simple cycles are non-negative \refCond{\CondNNSimpleCycles},
	and all words shorter than the number of states are accepting \refCond{\CondAllShortAccepting}.
    \item
	There exists an initial counter value $c_0\in\N$
	such that $\Lang{\initialstate,c_0}$ is universal if, and only if,
	all simple cycles are non-negative \refCond{\CondNNSimpleCycles},
	and the underlying automaton is universal \refCond{\CondAutUniversal}.
    \item For any $c_0\in\N$,
	there exists a bound $b\in\N$
	such that the bounded language $\bLang{\initialstate,c_0}{b}$ is universal if, and only if,
	\CondAllCyclesZero\ the effect of all simple cycles is $0$
	and
	\CondAllShortBAccepting\ all words shorter than the number of states are
	in $\bLang{\initialstate,c_0}{b'}$ for $b'\eqdef\card{\states}\cdot\norm{\transitions}$.
\end{enumerate}
\end{restatable}

\medskip
\noindent
The following is a direct consequence of 
\cref{lem:DOCN-conditions,lem:DOCN-universalities-char}.
\begin{theorem}
	\label{thm:DOCN-complexities}
	The universality, initial-value universality, and bounded universality problems
	for deterministic one-counter nets
	are in \NL\ assuming unary encoding,
	and 
        in \myNC{} assuming binary encoding.
\end{theorem}

For the special case of DOCN over single letter alphabets, it is possible to derive even better
upper bounds, based on the particular shape of the underlying automaton.

Recall that a deterministic automaton over a singleton alphabet is in the shape of a lasso:
it consists of an acyclic path that ends in a cycle.

\begin{restatable}{lemma}{detUnaryConditions}
	\label{lem:DOCN-conditions_singleton}
	For any given
	deterministic one-counter net 
	$\sys{A}=\ocntuple$ 
	with $\card{\alphabet}=1$ and $c_0,b\in\N$,
	one can verify in deterministic logspace (\L) that
	\refCond{\CondAutUniversal} the underlying DFA is universal.
	%
	Moreover, conditions
	\refCond{\CondAllShortAccepting},
	\refCond{\CondAllShortBAccepting},
	\refCond{\CondNNSimpleCycles},
	and \refCond{\CondAllCyclesZero}
	as defined in \cref{lem:DOCN-conditions}
	can be verified in \L\ assuming unary encodings
	and in \myNC{2}\ assuming binary encodings.
\end{restatable}

Using \cref{lem:DOCN-conditions_singleton} and the characterisation of the three universality problems by 
\cref{lem:DOCN-universalities-char},
we get the desired complexity upper bounds.

\begin{theorem}
	\label{lem:DOCN-singleton-complexities}
	The universality, initial-value universality, and bounded universality problems
	of deterministic one-counter nets over a singleton alphabet
	are in \L\ assuming unary encoding
	and 
	in \myNC{2}\ assuming binary encoding.
\end{theorem}

\section{Unambiguous Systems}
\label{sec:unambiguous}
In line with the usual definition of unambiguous finite automata,
we call an OCN with a given initial configuration \emph{unambiguous} iff for every word in its language there exists exactly one accepting run.
Since the language of an OCN depends in a monotone fashion on the initial counter value,
there is also a related, but different, notion of unambiguity.
We call an OCN (which has a fixed initial state $\initialstate$) \emph{structurally unambiguous} if 
the unambiguity condition holds for every initial counter $c_0$.
Notice that every OCN that has an unambiguous underlying automaton is necessarily structurally unambiguous.
We will show (\cref{lem:UOCA-structural-unambiguity}) that these conditions are in fact equivalent.

In \cite{czerwiski:hal-02483495}, the complexity of the universality problem 
for unambiguous vector addition systems with states (VASSs) 
was studied.
In particular, for unambiguous OCNs, it is shown that checking universality
is in \myNC{2} and \NL-hard, assuming unary encoded inputs,
and in \PSPACE\ and \coNP-hard, assuming binary encoding.
The special case of unambiguous OCN over a single letter alphabet is not considered there,
nor are the initial-counter -- and bounded universality problems.
We discuss these problems in the remainder of this section.

We assume w.l.o.g, that for any given OCN,
all states in the underlying automaton are reachable from the initial state,
and that from every state it is possible to reach an accepting state.
States that do not satisfy these properties can be removed in \NL.
Moreover, all algorithms we propose need to check universality for the underlying automaton,
and hence rely on the following computability result (see~\cite{WENGUEY199643} for a proof for general alphabet, and \cref{apx:lem:universality_of_UFA} for singleton alphabet).
\begin{restatable}{lemma}{UFAUniv}
    \label{lem:universality_of_UFA}
	Universality of an unambiguous finite automaton over single letter alphabet is in \NL, and over general alphabet is in \myNC{2}. 
\end{restatable}

We will start by considering the universality problem
for unambiguous OCNs over a single letter alphabet.
Here, unambiguity implies a strong restriction on accepting runs:
if a run is accepting then it contains at most one positive cycle (which may be iterated multiple times).

\begin{restatable}{lemma}{unambOneLoop}
    \label{lem:UOCA-single-loop}
	Let $\pi=\pi_1\pi_2\pi_3$ be an accepting run where $\pi_2$ is a positive simple cycle.
	Then 
	$\pi_3=\pi_2^k\pi_4$
	for some $k\in\N$ and acyclic path $\pi_4$.
\end{restatable}
\begin{proof}
	Assume towards contradiction that there is an accepting run
	$\pi=\pi_1\pi_2\pi_3\pi_4\pi_5$, where $\pi_2$ is a positive simple cycle
	and $\pi_4$ is a simple cycle.
	Based on this we show that the system cannot be unambiguous.
	Let $c=\card{\states}\cdot\norm{\transitions}$ and denote by $\len{\pi}$ the length of path $\pi$.
	
	Since $\pi_2$ has a positive effect, it follows that $\pi'=\pi_1\pi_2^{\len{\pi_4}+c\cdot \len{\pi_2}}\pi_3\pi_4\pi_5$ is an accepting run. But there is a second run that reads the same word, namely $\pi''=\pi_1\pi_2^{c\cdot \len{\pi_2}}\pi_3\pi_4^{\len{\pi_2}}\pi_5$. The second run is indeed a run
	as the increment along $ \pi_2^{c\cdot\len{\pi_2}}$ is bigger than any possible negative
	effect of $\pi_4^{\len{\pi_2}}$.
	Moreover the lengths of both runs are the same
	as $\pi_2^{\len{\pi_4}}=\pi_4^{\len{\pi_2}}$. 
\end{proof}

A consequence of \cref{lem:UOCA-single-loop} is that if along any accepting run the value of the counter exceeds
$\bound{0}=\card{\states}\cdot\norm{\transitions}$ then it cannot drop to zero afterwards, 
as it would require at least one negative cycle to do so.
One can therefore encode all counter values up to $\bound{0}$ into the
finite-state control and solve universality for the resulting UFA. \Cref{lem:universality_of_UFA} thus yields the following.

\begin{restatable}{theorem}{unambU}
    \label{thm:UOCA-universality-unary-single}
	The universality problem of unary encoded unambiguous one-counter nets over a singleton alphabet is in \NL.
\end{restatable}

We consider next the initial-value universality problem for unambiguous OCNs.
Since whether an OCN is unambiguous depends on the initial counter value,
the initial-value universality problem is only meaningful for structurally unambiguous
systems, those which are unambiguous regardless of the initial counter.
We first observe a simple fact about these definitions.

\begin{restatable}{lemma}{unambSU}
    \label{lem:UOCA-structural-unambiguity}
	An OCN is structurally unambiguous if and only if its underlying automaton is unambiguous.
\end{restatable}

\begin{restatable}{lemma}{unambSUstruct}
\label{lem:SUOBA-structure}
	Consider a structurally unambiguous OCN with initial state $\initialstate$.
	There exists an initial counter $c_0$ so that $\Lang{\initialstate,c_0}=\alphabet^*$
	if, and only if, the underlying automaton is universal and has no negative cycles.
\end{restatable}

The following is a direct consequence of
\cref{lem:SUOBA-structure}
and the complexity bounds
provided by
\cref{lem:universality_of_UFA,lem:DOCN-conditions}, for the cycle condition \refCond{\CondNNSimpleCycles}.
\begin{theorem}
	\label{thm:SUOCN-iv-universality}
	The initial-value universality problem of structurally unambiguous one-counter nets 
	is 
	in \myNC{2} assuming binary encoding,
	and 
	in \NL\ assuming unary encoding and single-letter alphabets.
\end{theorem}

Finally, we turn our attention to the bounded universality problem for unambiguous OCNs.
This turns out to be quite easy, due to the following observation.

\begin{restatable}{lemma}{unambBUnoLoops}
\label{lem:UOCN-no-positive-loops}
	If an unambiguous OCN is bounded universal then no accepting run contains a positive cycle.
\end{restatable}

\begin{restatable}{theorem}{unambBU}
	\label{thm:UOCN-bu}
	\label{thm:UOCN-bu-bin}
	The bounded universality problem of unambiguous one-counter nets with unary-encoded transition weights
	is in \myNC{2}, and in \NL\ if the alphabet has only one letter, and for binary-encoded transition weights it is in \PSPACE.
\end{restatable}

\newpage
\bibliography{autocleaned}

\newpage
\appendix

\section{Proofs of \cref{sec:nondet}}
\label{apx:thm:Undecidable1}

\undecBU*
\begin{proof}
\label{apx:thm:Undecidable2}
	We show that a given two-counter machine $\M$ halts if and only if the corresponding one-counter net $\A'$, as constructed in \cref{sec:NondetBounded}, is bounded universal for an initial counter value $0$.
	
	\subproof $\Rightarrow:$\emph{If $\M$ halts}, its (legal) run has some length $n-1$. We claim that $\A'$ is universal with the counter bound $2n$.
	
	Consider some word $w'$ over the alphabet of $\A'$. We shall describe an accepting run $\rho'$ of $\A'$ on $w'$.
	In the first $n$ steps, $\rho'$ remains in $q'_0$, increasing the counter to $n$. Then, it moves to $q_0$. 
	In the rest of the run, $\rho'$ continues as the accepting run $\rho$ of $\A$ on the word $w$ that is the suffix of $w'$ from the $n+1$ position (as described in the proof of \cref{thm:Undecidable1}), except for the following changes: whenever it is in $q_0$ and the counter is bigger than $n$, it goes to $q_7$ on $\#$. In $q_7$, it uses the self loop until the counter's value becomes $n$ and then goes to $q_0$. 
	
	If the length of $w'$ is up to $n$, then $\rho'$ is obviously accepting, as it remains in the accepting states $q'_0$ and $q_0$, and the counter need not exceed $2n$ nor go below $0$.
	
	If the length of $w'$ is more than $n$, we prove that for every segment between two consequent $\#$'s, as well as the segment after the last $\#$, the run $\rho'$ may either reach $heaven$ or reach $q_0$ with counter value at least $n$, and proceed from $q_0$ to $q_1..q_6$ with counter value exactly $n$. This will immediately imply that $\rho'$ is accepting.
	
	The challenge is to show that the counter of $\A'$ never needs to exceed $2n$. (It does not go below $0$, since we go from $q_0$ to $q_1..q_6$ with a counter value of at least $n$ (in this case exactly $n$), which satisfies the assumptions in the proof of \cref{thm:Undecidable1}.)
	
	Now, in states $q_1, q_2, q_4, q_6$, and $q_7$ there is no problem, as the counter never gets above its value when entering these states. Yet, in states $q_3$ and $q_5$ there is a potential problem, since $\A'$'s counter increases when $\M$'s counters increase. However, since the (legal) run of $\M$ is of length $n-1$, a violation must occur within up to $n$ steps. Hence, getting to states $q_3$ and $q_5$ with counter value of exactly $n$, the run $\rho'$ may return to $q_0$ over the first violation, and thus need not increase the counter's value to more than $2n$. Observe that when returning to $q_0$ the counter's value might be bigger than $n$, in which case $\rho'$ will later decrease it to exactly $n$ by going to $q_7$.

	\subproof $\Leftarrow:$ \emph{If $\M$ does not halt}, for every positive integer $n$, we build the word $w'_n$ and show that it is not accepted by $\A'$ for an initial counter value $0$ and a bound $n$ on the counter.
	
	The word $w'_n$ consists of $n+2$ segments between $\#$'s, where each segment is the prefix of length $n$ of the (legal) run of $\M$.
	Consider the possible runs of $\A'$ on $w'_n$. In $q'_0$ it can stay up to $n$ steps, entering $q_0$ with a counter value of up to $n$. Then it should accept from $q_0$ the suffix of $w'_n$, which contains $n+1$ segments as described above. However, as shown in the proof of \cref{thm:Undecidable1}, using all states except for $q_7$, it must decrease the counter value in each segment, and so is the case if using $q_7$. Hence, the run must stop after at most $n$ segments and cannot be accepting.
\end{proof}

\newpage
\section{Proofs of \cref{sec:unary}}
\unaryLPF*
\begin{proof}
	Let $n\in \Lang{\initialstate,c_0}$, and let $\pi=\initialstate,s_1,\ldots,s_n$ be an accepting run of the OCN on $n$. For each state $q$ visited by $\pi$, let $\first(q)$ and $\last(q)$ denote the first and last indices where $q$ occurs in $\pi$, respectively. Let $\Marks\eqdef\{\first(q),\last(q) : q\mbox{ occurs in }\pi\}$ be the set of all markings in $\pi$. Observe that $|\Marks|\le 2|\states|$. 
	
	We reshape $\pi$ into linear form in two phases. In the first phase, we move cycles around such that in the obtained path, any infix between two marked positions consists of a simple path, and a collection of simple cycles. In the second phase, we  replace most of the simple cycles, such that any infix between two marked positions
	consists of a relatively short path, and a single repeating cycle (which completes the linear form).
	Crucially, in both phases we must take care that the path remains executable. 
	The crux of the proof is that instead of simply shifting cycles, we also change their starting point, such that they always start from a nadir, thus making them executable with any counter value.
	
	For the first phase, consider an interval $[i,i+|\states|]$ in $\pi$ that does not intersect $\Marks$ (if no such interval exists, we proceed to the second phase). Since this interval has $|\states|+1$ states, it contains some simple cycle $\gamma=x_1,x_2,\ldots,x_k$. Let $d$ be a nadir of $\gamma$, and observe that necessarily $\first(x_d)<i$ and $\last(x_d)>i+|\states|$, since the interval $[i,i+|\states|]$ does not contain any marks.
	
	We now split into two cases. 
	\begin{itemize}
		\item If $\effect{\gamma}\ge 0$, we modify $\pi$ by removing the cycle $\gamma$ from the interval $[i,i+|\states|]$, and instead adding the shifted cycle $\gamma^{\leftarrow d}$ at index $\first(x_d)$. 
		
		Observe that the modified path is still executable, since by \cref{obs:shifted cycle} the cycle $\gamma^{\leftarrow d}$ is good, and can be executed with any counter value, and following its execution, the remaining path either has higher counters (up to where $\gamma$ occurred) or the same values as in $\pi$ (after where $\gamma$ occurred).
		
		\item If $\effect{\gamma}< 0$, we modify $\pi$ by removing the cycle $\gamma$ from the interval $[i,i+|\states|]$, and instead adding the shifted cycle $\gamma^{\leftarrow d}$ at index $\last(x_d)$. 
		
		Observe that the modified path is still executable. Indeed, by \cref{obs:shifted cycle} $\effect{\gamma^{\leftarrow d}}=-\depth{\gamma^{\leftarrow d}}$, and so $\gamma^{\leftarrow d}$ can be executed
		as long as the counter is at least $\effect{\gamma^{\leftarrow d}}$.
		Moreover, removing this negative cycle results in a
		run in which, all counter-values from the index of removal are increased by $-\effect{\gamma}$.
		In particular, at index $\last(x_d)$ it is at least $0+\effect{\gamma^{\leftarrow d}}$, so $\gamma^{\leftarrow d}$ can be executed. Notice that moving a negative cycle like this results in
		a path that is executable an has the same effect as $\pi$.
	\end{itemize}
	This completes the first phase. We remark that conceptually, this cycle modification takes place in a single ``shot'' for all cycles, so that the indices in $\Marks$ do not change after every cycle is moved, but are rather the same for all cycles being moved (otherwise intervals may ``expand'', and $\Marks$ becomes ill-defined).
	
	\medskip
	We now proceed to the second phase. Let $\pi'$ be the path obtained after the first phase. We refer to any cycle that was moved in $\pi$ as a \emph{dangling cycle}. Thus, $\pi'$ consists of at most $2|\states|$ intervals\footnote{The first and last indices of $\pi$ must be marked and so there are in fact at most $2|\states|-1$ intervals.} that contain no non-dangling cycles, and at most $2|\states|$ indices on which there are dangling cycles (namely the indices in $\Marks$). Furthermore, the dangling cycles always start at their respective nadirs. 
	
	We now proceed to eliminate most dangling cycles at each state. Consider some mark $\first(q)$ or $\last(q)$ in $\Marks$. For each $1\le t\le |\states|$, consider all simple cycles of length $t$ where $q$ is a nadir, and let $\mu_{q,t}$ be such a cycle of maximal effect. We now replace every dangling cycle of length $t$ in $\first(q)$ with $\mu_{q,t}$. Clearly the effect of the cycles does not decrease, so the path remains executable. Furthermore, we maintain the length of the paths, so the path still represents a run on $n$.
	
	Finally, within each mark, we can bunch the cycles by length, so that all cycles of the same length are executed consecutively. Thus, the obtained path consists of at most $2|\states|$ simple paths and $2|\states|\cdot |\states|=2|\states|^2$ simple cycles, which is a linear form as required.
\end{proof}

\unaryLPFpump*
\begin{proof}
	Let $\gamma_r$ be a shortest good cycle on $r$, and 
	let $\rho$ be a an accepting run that passes through $r$. We write $\rho=\pi_1,r,\pi_2$, where $\pi_r$ is a prefix of the run before it visits $r$ and $\pi_2$ is the suffix after visiting $r$ (note that $r$ may occur in $\pi_2$). Furthermore, by \cref{lem:linear form of paths} we can assume $\pi_1$ and $\pi_2$ are in linear form of length at most $2|Q|^2$. 
	Thus, we can write $\pi_1=\tau_0\gamma_1^{e_1}\tau_1\cdots \tau_{k-1}{\gamma_k}^{e_k}\tau_k$ with $k\le 2|Q|^2$, and similarly for $\pi_2$. 
	
	We now start by replacing negative cycles in $\pi_1$ and in $\pi_2$ by repetitions of $\gamma_r$ (the good cycle on $r$). This is done as follows. For every subset of cycles whose combined length equals $m|\gamma_r|$ for some $m\in \N$, we remove those cycles and replace them by $m$ iterations of the good cycle $\gamma_r$. Since we only remove negative cycles, and since $\gamma_r$ has non-negative effect and depth $0$, the run remains executable. Recall that the $\gamma_i$ cycles are simple, and are therefore of length at most $|Q|$. Thus, after removing cycles in this manner, we are left with at most $|\gamma_r|-1\le |\states|$ negative cycles of every length. 
	
	We now aim to remove non-negative cycles in the same fashion. This, however, requires some caution, as some cycles might have effect greater than that of $\gamma_r$, or appear before the run visits state $r$ for the first time,
	and therefore replacing them with $\gamma_r$ may cause the path to become non-executable. 
	Recall that by \cref{def:linear path scheme} (and indeed, by the construction in the proof of \cref{lem:linear form of paths})
	all the non-negative $\gamma_i$ cycles start from their nadir, and therefore have depth $0$. In addition, after removing the negative cycles as done above, the path length (excluding the non-negative cycles) is at most $2|\states|^2+|\states|^2=3|\states|^2$ in each of $\pi_1$ and $\pi_2$. Thus, the maximal depth possible along the entire path is $6|\states|^2\norm{\transitions}$. Thus, as long as a (strictly) positive cycle (or a combination thereof) is taken enough times to maintain the counter above $6|\states|^2\norm{\transitions}$, the path remains executable. We can now proceed to replace non-negative cycles with $\gamma_r$ in the same manner done for negative cycles, while maintaining executability.
	
	We thus end up with a modified run of the form $\eta_1 \gamma_r^{t} \eta_2$ where $\eta_1$ and $\eta_2$ are of length $\poly(|\states|,\norm{\transitions})$, which implies the claim.
\end{proof}

\unaryPumpDFA*
\begin{proof}
	From \cref{lem:linear form with pump state} it follows that there exists a bound $\bound{1}\in \poly(|\states|,\norm{\transitions})$ such that every word accepted with a run that goes through $r$ is of the form $x+y|\gamma_r|$ where $x,y\in \+N$ and $x\le \bound{0}$. Thus, we can construct a DFA of size $\bound{2}\eqdef\bound{1}+|\gamma_r|$ whose form is an initial prefix of length $\bound{1}$, followed by a cycle of length $|\gamma_r|$, and whose accepting states correspond to all the $x$ above, with corresponding accepting states on the cycle.
\end{proof}

\unaryWitnessBound*
\begin{proof}
	Let $\bound{2}$ be as per \cref{lem:pump state DFA}, and define $\bound{3}\eqdef\bound{2}^{|\Pump|}\le \bound{2}^{|\states|}$. 
	Observe that by taking the product of the DFAs obtained in \cref{lem:pump state DFA}, we can construct a DFA $\?D$ of size at most $\bound{3}$ for $\N\setminus \PS$. Then, $\N\setminus \PS$ is infinite iff there exists a word of length $\bound{3}\le n\le 2\bound{3}$ that is accepted by $\?D$ (as such a word is necessarily accepted by a run that contains a cycle in $\?D$).
	
	Towards the claim, if $\N\setminus \PS$ is infinite, then $\Lang{\initialstate,c_0}\neq \N$, and clearly if there exists $n<\bound{2}$ such that $n\notin \Lang{\initialstate,c_0}$ then again, $\Lang{\initialstate,c_0}\neq \N$.
	
	Conversely, assume $\Lang{\initialstate,c_0}\neq \N$. We claim that either there exists $n<\bound{2}$ with $n\notin \Lang{\initialstate,c_0}$, or $\N\setminus \PS$ is infinite. Indeed, observe that since $\?D$ is obtained as the product of singleton-alphabet DFAs, then it has a ``lasso'' shape: a finite prefix of states, followed by a cycle. Moreover, the size of the prefix is at most $\bound{2}$, namely the maximal size of the prefix in each of the DFAs in the product. Thus, if there exists $n<\bound{2}$ with   $n\notin \Lang{\initialstate,c_0}$ then we are done, and otherwise there is some $n>\bound{2}$ with $n\notin \Lang{\initialstate,c_0}$, and in particular $n\notin \PS$, so $\?D$ accepts some word along its cycle, and so accepts infinitely many words, and in particular some word $\bound{3}\le n\le 2\bound{3}$.
\end{proof}

\unaryCheckLPS*
\begin{proof}
	Checking that the transitions follow those of the OCN can be done in polynomial time, since we only need to check the underlying path, regardless of the exponents.
	In order to check that the counter value remains non-negative, we observe that for any cycle $\gamma_i$, if $\effect{\gamma_i}\ge 0$, then $\gamma_i$ is taken from a nadir (by \cref{def:linear path scheme}), and hence can be taken with any counter value.
	If that is the case, then we can compute directly $\effect{\gamma_i^{e_i}}=e_i \cdot\effect{\gamma_i}$. 
	Otherwise, if $\effect{\gamma_i}< 0$, then in order to check if $\gamma_i^{e_i}$ is executable from counter value $c$, it suffices to check that $(e_i-1)\cdot\effect{\gamma_i}-\depth{\gamma_i}\le c$. Indeed, for negative cycles, the last iteration is the ``hardest''.
	Again, we can now compute $\effect{\gamma_i^{e_i}}=e_i\cdot \effect{\gamma_i}$. 
	
	Thus, we can keep track of the counter value along the underlying path, and update it directly for every cycle. This takes polynomial time overall.
\end{proof}

\unaryIVchar*
\begin{proof}
	For the first direction, assume $\Lang[\?A]{\initialstate,c_0}=\N$ for some $c_0$. Clearly $\Lang[\?N]{\initialstate}=\N$
	as otherwise some word is not accepted in the underlying NFA, let alone the OCN. Assume by way of contradiction that $\N\setminus \PS$ is infinite,
	and recall that in every accepting run on a word $n\in \N\setminus \PS$, all cycles must be negative.
	Thus, for long enough words, the counter value, starting at $c_0$, must become negative, which is a contradiction. 
	
	Conversely, if $\N\setminus \PS$ is finite and $\Lang[\?N]{\initialstate}=\N$, we can take an initial counter value large enough so that all words not in $\PS$ have accepting runs. Then, similarly to \cref{lem:linear form with pump state}, we can show that every word in $\PS$ has an accepting run of the form $\tau_1\gamma_r^{t}\tau_2$ with $\tau_1$ and $\tau_2$ of length $\poly(|\states|)$ and where $\gamma_r$ is the canonical good cycle from state $r\in\Pump$
	with maximal effect.
	Notice here that the bound on the lengths of paths $\tau_1$ and $\tau_2$ is polynomial only in the number of states and not, as in \cref{lem:linear form with pump state}, also in $\norm{\transitions}$.
	This is because we can safely remove any combination of simple cycles in these sub-paths
	without preserving the executability of the resulting path in the net.
	A large enough counter value ensures that the prefix and suffix are executable, so all words in $\PS$ are accepted as well.
\end{proof}

\unaryBUstable*
\begin{proof}
	If $q$ is at the nadir of a simple zero cycle, then $\card{\states}$ bounds its length and we are done.
	
	Otherwise, since $q$ admits a negative cycle, then there is a state $x\in Q$ that admits a simple negative cycle $\gamma$ such that $x$ and $q$ are reachable from each other. Let $\tau_1$ and $\tau_2$ be simple paths from $q$ to $x$ and from $x$ to $q$, respectively. Let $s=\effect{\tau_1\tau_2}+1$, then $\chi=\tau_1\gamma^s\tau_2$ is a negative cycle of length at most $3|\states|\cdot \norm{\transitions}$. 
	
	Let $\eta$ be a simple positive cycle that has a nadir at $q$.	
	Then $q$ admits the zero cycle $\zeta_q=\eta^{-\effect{\chi}}\cdot \chi^{\effect{\eta}}$
	and $\bound{5}\eqdef
	\card{\states}  
	\cdot (\card{\states}\cdot\norm{\transitions}) 
	+ 
	(3\card{\states}\cdot\norm{\transitions}) 
	\cdot \card{\states} 
	$
	satisfies the claim.
\end{proof}

\unaryBUchar*
\begin{proof}
	By \cref{cor:singleton alphabet bounded-univ bounded langvia}, there exists a bound $\bound{7}$ such that all words in $\?S$ are accepted with paths whose counter values remains below $\bound{7}$. Hence, if there are only finitely many words that are outside $\?S$, and $\Lang[\?N]{c_0}=\N$, then the counter values among the runs on the remaining finite set of words are clearly bounded. Hence, $\Lang{\initialstate,c_0}$ is bounded-universal.
	
	Conversely, assume $\N\setminus \?S$ is infinite, we show that $\Lang{\initialstate,c_0}$ is not bounded-universal. First, if $\Lang[\?N]{\initialstate}\neq \N$ the OCN cannot be universal, and in particular it is not bounded-universal. Observe that by \cref{def:stable states}, words outside $\?S$ can be accepted only with paths on which the number of alternations between positive and negative cycles is at most $|\states|$, and that do not contain zero cycles. Since only finitely many words can be accepted using a bounded number of positive cycles, it follows that if $\N\setminus\?S$ is infinite, then for every $M\in\N$ there exists a word that is only accepted by runs that have a positive cycle taken at least $M$ times, and hence have effect at least $M$. 
	It follows that $\Lang{\initialstate,c_0}$ is not bounded-universal.
\end{proof}

\unaryUnivBinary*
\begin{proof}
	Following our algorithmic scheme, an $\NP^\NP$ algorithm for non-universality proceeds as follows. non-deterministically either (1) guess $n<\bound{3}$, and check (using an $\NP$ oracle as per \cref{lem:checking linear path scheme}) that $n\notin \Lang{\initialstate,c_0}$, or (2) guess $\bound{3}^{|\states|}\le n\le 2\bound{3}^{|\states|}$ and check that $n\notin \Lang[\?A^r]{\initialstate,c_0}$ for all $r\in \Pump$, using $|\states|$ calls to an $\NP$ oracle as per \cref{lem:checking linear path scheme}. 
\end{proof}

\newpage
\section{Proofs of \cref{sec:deterministic}}
\label{apx:lem:binary-addition}
\detBinaryAddition*
\begin{proof}
	Addition of two integers written in binary can be done in $\myAC{0}$ \cite{IntroductionToCircuitComplexity}, and therefore in $\myNC{1}$.
	As the summation of $n$ numbers can be done in $\log n$ iterations (whereby each iteration reduces the number of elements  by a factor of $2$ by adding up   $\alpha_{2i}$ and $\alpha_{2i+1}$, for every index $i$ up to half the number of elements), and each iteration is in $\myNC{1}$ (by performing in parallel all of these additions), we get that the overall problem is in \myNC{2}.
\end{proof}

\detConditions*
\begin{proof}
\label{apx:lem:DOCN-conditions}
	\textbf{Unary encoding}.
	All conditions can be shown to be in \NL\ 
	using the theorems of Savitch (reachability in finite directed graphs is in \NL)
	and Immerman–Szelepcsényi ($\NL=\coNL$).
	Indeed, \refCond{\CondAutUniversal} holds
	iff no non-accepting state is reachable in the underlying automaton.
	For the remaining conditions, just notice that the assumption that inputs are given in unary
	means that all relevant numbers are bounded polynomially in the input.
	For instance, to show that \refCond{\CondNNSimpleCycles} does not hold,
	one simply guesses the offending simple cycle and stepwise computes its effect in binary representation.
	
	\textbf{Binary encoding}.
	Let's first consider
	condition \refCond{\CondAllShortAccepting}.
	This fails iff
	there is a short word whose run in $\sys{A}$
	either ends in a non-accepting state or reduces the counter below zero.
	The first case is again a simple reachability condition in the underlying DFA.
	The second case reduces to a coverability problem as follows.
	
	For $k\in\N$, let $\sys{A}\x k \eqdef (\states\x\{0,1,\ldots,k\},\alphabet,\transitions',\fstates',\initialstate')$
	be the OCN that results from $\sys{A}$ by adding a step-counter up to $k$ into the
	states.
	That is,
	$\transitions'\eqdef \{((p,i),\alpha, e, (q,i+1)): (p,\alpha,e,q)\in \transitions, i\le k\}$,
	$\fstates'\eqdef \fstates\x\{0\ldots k\}$, and 
	$\initialstate'\eqdef (\initialstate,0)$.
	Further, let $\sys{B}$ denote the OCN $\sys{A}\x\card{\states}$, in which all transition effects are inverted.
	Notice that for every word $w$ of length $\len{w}\le \card{\states}$,
	the effect of its induced run in $\sys{A}$ (and $\sys{B}$) is between
	$-\bound{}$ and $\bound{}$, for $\bound{}\eqdef \card{\states}\cdot\norm{\transitions}$.
	Such a word cannot be accepted by $\sys{A}$ from $(\initialstate,c_0)$ iff
	the run it induces in $\sys{B}$ starting from
	$(\initialstate', \bound{})$ leads to some configuration $((q,\len{w}),(\bound{}+c_0+1))$.
	This reachability question about $\sys{B}$ can be answered in \myNC{} \cite[Lemma 1 and Theorem 15 ]{almagor2019coverability}, and since $\sys{A}$  and $\sys{B}$ are of polynomially the same size, also in \myNC{} with respect to $\sys{A}$.
	
	An \NC\ upper bound for condition \refCond{\CondAllShortBAccepting} is completely analogous and differs only in that an additional reachability check should be taken, in which the weights in $\sys{B}$ are not inverted and the target configuration is $((q,\len{w}),(\bound{}+b-c_0+1))$.
	
	\smallskip
	Conditions 
	\refCond{\CondNNSimpleCycles} and \refCond{\CondAllCyclesZero}
	on the effect of simple cycles can be verified in \NC\, by a similar reduction to coverability.
	For example, to check if a simple cycle with negative effect exists it suffices
	to check that it is possible in $\sys{B}$ to start in a configuration $((q,0),\bound{})$
	and cover a configuration $((q,k),(\bound{}+1))$ for some $0< k < |Q|$.
	
	We can do slightly better than that and check these conditions in \myNC{2}, as follows.
	Let $Q=\{p_1,p_2, \ldots, p_{|Q|}\}$, and for every $0< k < \card{\states}$,
	let $M_k$ denote the $\card{\states}\times \card{\states}$ matrix of elements in $\Z\cup\infty$, where the entry for $i,j$
	equals the minimal effect of a path of length $k$ from state $p_i$ to $p_j$.
	Then, $M_k$ can be computed in \myNC{2} using standard repeated-squaring in the min-plus semiring~\cite{aho1974design}
	
	
	To check condition \refCond{\CondNNSimpleCycles},
	that all simple  cycles have non-negative effect, we just need to check (in parallel)
	that all entries in the main diagonal of all the $M_k$ matrices are non-negative.
	The same procedure, applied to an OCN that is derived from  $\sys{A}$  by inverting all transition weights,
	allows to check for the presence of positive simple cycles, and hence for an \myNC{2} algorithm to check 
	condition \refCond{\CondAllCyclesZero}.
\end{proof}

\detUs*
\begin{proof}
\begin{enumerate}
    \item (Normal Universality):
	\label{apx:lem:DOCN-universality}
	Clearly both conditions are necessary for the system to be universal.
	To see why they are sufficient for universality,
	assume that \refCond{\CondNNSimpleCycles} holds
	and consider shortest word $w\not\in\Lang{\initialstate,c_0}$.
	Then the run on $w$ cannot contain 
	any non-negative cycle because this would contradict the minimality assumption. 
	Since we assume \refCond{\CondNNSimpleCycles}, that all cycles are non-negative, the run on $w$ must have no cycles.
	Thus, $\len{w}\le \card{\states}$ which is impossible due to \refCond{\CondAllShortAccepting}.

    \item (Initial-Value Universality):
	\label{apx:lem:DOCN-iv-universality}
	If both conditions hold then 
	any cycle on any run must have non-negative effect.
	So if one picks $c_0\eqdef \card{\states}\cdot\norm{\transitions}$
	then the counter cannot become negative on any run
	and the language $\Lang{\initialstate,c_0}$ equals that of the underlying automaton,
	namely $\alphabet^*$ by condition \refCond{\CondAutUniversal}.
	
	Conversely, since $\Lang{\initialstate,c_0}$ is always included in the language of the underlying automaton,
	condition \CondAutUniversal~is clearly necessary. If \CondNNSimpleCycles\ fails
	then, because the system is deterministic, for every number $c_0$ there must be a word $w(c_0)\in\alphabet^*$ whose run has an effect strictly below $-c_0$. Then $w\notin\Lang{\initialstate,c_0}$.
	Therefore both conditions are necessary.
    \item (Bounded Universality):
	\label{apx:lem:DOCN-bounded-universality}
	Trivially, both conditions are necessary. 
	For the opposite direction, assume that the conditions hold.
	We contradict the assumption that $\bLang{\initialstate,c_0}{b'}\neq\alphabet^*$. 
	If that was the case, we can pick a shortest word $w$ not in that language.
	The run of this word cannot contain a cycle, because by condition \CondAllCyclesZero\
	all cycles have zero effect on the counter and therefore the presence of a cycle on the run would contradict the assumed minimality of $\len{w}$.
	This implies that $w$ is no longer than the number of states,
	and by condition \CondAllShortBAccepting\ it must be in $\bLang{\initialstate,c_0}{b'}$. Contradiction.\qedhere
\end{enumerate}
\end{proof}

\detUnaryConditions*
\begin{proof}
	\label{apx:lem:DOCN-conditions_singleton}
	Condition \refCond{\CondAutUniversal} is equivalent to checking that all states are accepting ($\states=\fstates$).
	For the other conditions,
	notice that if all numbers are encoded in unary then one only needs to compute
	numbers bounded polynomially in $\card{Q}$ and $\norm{\transitions}$. This can be done in deterministic logspace
	by representing them in binary.
	If numbers are already encoded in binary then the \myNC{2} bounds 
	follow from \cref{lem:binary-addition}.
\end{proof}

\section{Proofs of \cref{sec:unambiguous}}

\UFAUniv*
\begin{proof}
	\label{apx:lem:universality_of_UFA}
	The lemma was proven in \cite{WENGUEY199643}, for the general alphabet.
	For the single letter alphabet we have that if the language is not universal then the shortest not accepted word is bounded by $|Q|$
	\cite{DBLP:conf/dcfs/Colcombet15} (Lemma 2). Thus to verify universality, we need to test if for every $0\leq i\leq |Q|$ there is an accepting run
	of length $i$, which can be tested in $\NL$.
\end{proof}
	
\unambU*
\begin{proof}
\label{apx:thm:UOCA-universality-unary-single}
	By \cref{lem:UOCA-single-loop} it is possible to construct an unambiguous finite automaton (UFA) of polynomial
	size, which is universal if and only if the net is universal. 
	This can be done by bounding the counter from above by $\bound{0}$,
	remembering its value in the states,
	and switching to a copy of the underlying automaton
	once the counter is observed to exceed this bound.
	It is easy to see that every run in the net induces a run in the automaton and vice-versa. 
	The number of states of this new finite automaton is $\card{\states}\cdot (1+\bound{0}) + \card{\states}$.
	Since the constructed UFA is still over a single letter alphabet,
	we can check if it is universal \NL\ by \cref{lem:universality_of_UFA}. 
\end{proof}

\unambSU*
\begin{proof}
	\label{apx:lem:UOCA-structural-unambiguity}
	If the underlying automaton is unambiguous then the net is as well, as every run of the net is also a run of the automaton.
	
	In the opposite direction, suppose that the underlying automaton is not unambiguous, then there is a word $w$ read by two accepting runs $\pi_1$ and $\pi_2$.
	If we start with the counter value bigger than $(\len{\pi_1}+\len{\pi_2}) \cdot \norm{\transitions}$ then the both runs in the underlying automaton will describe two different accepting runs in the OCN.
\end{proof}

\unambSUstruct*
\begin{proof}
\label{apx:lem:SUOBA-structure}
	\emph{``If''}.
	If all cycles have non-negative effect then an initial value of $c_0\eqdef\bound{0}$
	suffices to ensure that no run can drop the counter below zero.
	Consequently, the system behaves just like its underlying automaton, which is universal by assumption.
	
	\emph{``Only if''}.
	The language $\Lang{\initialstate}$ of the underlying automaton clearly includes $\Lang{\initialstate,c}$
	for any value $c\in\N$. By assumption that 
	there is $c_0$ with $\Lang{\initialstate,c_0}=\alphabet^*$,
	the underlying automaton must be universal.
	
	It remains to show that it cannot contain any (reachable) simple cycles with negative effect.
	Towards a contradiction, suppose that $\pi_1\pi_2\pi_3$ is an accepting run from a configuration $(\initialstate, c_0)$
	and that $\effect{\pi_2}<0$.
	Then there is must exist $k\in \N$ such that $\pi_1\pi_2^k\pi_3$ is not a run from the configuration $(\initialstate,c_0)$, as the counter runs out.
	By assumption, that the language of the net with initial configuration $(\initialstate,c_0)$
	is universal, there must be another run $\pi_4$ on the same word,
	and which is accepting.
	But now both runs, $\pi_4$ and $\pi_1\pi_2^k\pi_3$, are accepting from the configuration $(\initialstate,c_0+\norm{\transitions}\cdot \len{\pi_2}\cdot k)$ as the effect of 
	$\pi_2^k$ is larger than $\norm{\transitions}\cdot \len{\pi_2}\cdot k$.
	This means that the net is not structurally unambiguous, which contradicts our assumptions.
\end{proof}

\unambBUnoLoops*
\begin{proof}
	Suppose otherwise, then for any bound $k$ there will be an accepting run which is going 
	through configurations with counter value bigger than $k$, and from unambiguity, there is no other run that stays below the bound.  
\end{proof}

\unambBU*
\begin{proof}
	\label{apx:thm:UOCN-bu}
	\textbf{Unary encoded transitions:}
	By \cref{lem:UOCN-no-positive-loops}, if the OCN is bounded universal
	then every accepting run will only visit counter values below
	$\bound{1}\eqdef c_0+\bound{0}=c_0+\card{\states}\cdot\norm{\transitions}$.
	This means that the OCN is bounded universal if, and only if,
	$\bLang{\initialstate,c_0}{\bound{1}}=\alphabet^*$.
	This can be verified by checking universality for the UFA
	that results by remembering all bounded counter values in the finite state space.
	The claim now follows by \cref{lem:universality_of_UFA}.

	\textbf{Binary encoded transitions:}
	By \cref{lem:UOCN-no-positive-loops}, if the OCN is bounded universal
	then every accepting run will only visit counter values below
	$\bound{1}\eqdef c_0+\bound{0}=c_0+\card{\states}\cdot\norm{\transitions}$.
	This means that the OCN is bounded universal if, and only if,
	$\bLang{s_0,c_0}{\bound{1}}=\alphabet^*$.
	This can be verified by checking universality for the UFA
	that results by remembering all bounded counter values in the finite state space.
	The claim now follows by \cref{lem:universality_of_UFA} and the following fact \myNC{}$=PolyLog$ applied to the
	UFA which is of exponential size.
\end{proof}
\end{document}